\documentclass[letterpaper, 10 pt, conference]{ieeeconf}  
\usepackage[utf8]{inputenc}
\usepackage{amsfonts}
\usepackage{amsmath}
\usepackage{comment}
\usepackage{xcolor}
\usepackage{algorithm, algpseudocode}
\usepackage{subcaption}

\usepackage{pgfplots}
\pgfplotsset{compat=newest}
\usetikzlibrary{plotmarks}
\usetikzlibrary{arrows.meta}
\usepgfplotslibrary{patchplots}
\usepackage{grffile}
\usepackage{balance}

\newcommand{\R}{\mathbb{R}}
\newcommand{\bmat}[1]{\begin{bmatrix}#1\end{bmatrix}}

\newtheorem{theorem}{Theorem}
\newtheorem{proposition}{Proposition}
\newtheorem{remark}{Remark}
\newtheorem{definition}{Definition}

\newtheorem{assumption}{Assumption}

\IEEEoverridecommandlockouts                              
\makeatletter
\let\NAT@parse\undefined
\makeatother
\usepackage{hyperref}
\title{\LARGE \bf Optimization Based Planner--Tracker Design for Safety Guarantees}

\author{He Yin$^\star$, Monimoy Bujarbaruah$^\star$, Murat Arcak, and Andrew Packard  
\thanks{$^\star$ authors contributed equally} %
\thanks{E-mails: \tt\scriptsize{\{he\_yin, monimoyb, arcak, apackard\}@berkeley.edu.}} }

\begin{document}

\maketitle
  \thispagestyle{empty}
\pagestyle{empty}

\begin{abstract}
We present a safe-by-design approach to path planning and control for nonlinear systems. 
The planner uses a low fidelity model of the plant to compute reference trajectories by solving an MPC problem, while the plant being controlled utilizes a feedback control law that tracks those trajectories with an upper-bound on the tracking error. Our main goal is to allow for maximum permissiveness (that is, room for constraint feasibility) of the planner, while maintaining safety after accounting for the tracking error bound. We achieve this by parametrizing the state and input constraints imposed on the planner and deriving corresponding parametrized tracking control laws and tracking error bounds, which are computed \emph{offline} through Sum-of-Squares programming. The parameters are then optimally chosen to maximize planner permissiveness, while guaranteeing safety.
\end{abstract}

\section{Introduction}
Path planning and control of automated systems is a highly researched topic and a number of approaches exist to tackle this problem \cite{lavalle2006planning, frazzoliSurvey, gonzalez2015review}. A widely used approach for path planning and control is Model Predictive Control (MPC) \cite{rawlings2009model, borrelli2017predictive, kouvaritakis2016model}, where a model is used to predict system states over a finite horizon, and a sequence of optimal inputs is synthesized by solving a constrained finite time optimization problem minimizing a suitably chosen cost function. The first optimal input is applied to the system, and then the process is repeated, thus resulting in a receding horizon control strategy. If the ``planning'' model used for MPC predictions and the plant have no discrepancy, then the so called recursive feasibility, as well as stability of such an MPC controller are ensured by suitably choosing ``terminal conditions" in the MPC problem \cite[Chapter~12]{borrelli2017predictive}.
Such feasibility certificates are crucial for safety critical applications, where constraint violations are intolerable at any time during operation. However, under mismatch of planning model and the plant, the MPC optimization problem must be robustified.

Feasibility and stability properties of robust MPC have been studied in detail over the past few decades \cite{kothare1996robust, mayne2000constrained}. For linear systems, Tube MPC \cite{Langson2004RobustMP, Rakovic2012HomotheticTM} is a widely used approach that solves a computationally efficient convex optimization problem for robust control synthesis. 
Although Tube MPC design with feasibility and stability properties are proposed for nonlinear systems in \cite{kohler2019computationally, kollerFelix1}, the control synthesis problem becomes computationally demanding, due to non-convexity of the resulting optimization problem \cite{Boyd:2004:CO:993483}.  

To alleviate this issue, a typical approach in the path planning community is a two layer control architecture of planner--tracker design \cite{borrelliStefano, herbert2017fastrack, summetCT, kousik2018bridging, Singh2018RobustTW, YinStan2019}.
The high-level planning controller is synthesized and implemented online using a low-fidelity planning model, and imposes appropriately chosen constraints on the variables of the planner. 
The low-level tracking controller applied to the plant (referred to as the tracker) simultaneously ensures robust constraint satisfaction in closed loop 
by bounding the tracking error in a set.
The tracking controller is computationally expensive to synthesize, but it is typically a state feedback policy, making it cheaper to implement. Hence, it is synthesized offline a-priori, and the policy is invoked during run-time. 
However, if the constraints imposed on the planner are not chosen appropriately, either the planner can become infeasible, or the error bound can violate tolerable limits, resulting in the tracker violating safety constraints. 


In this paper we propose an optimization based approach to designing a path planning--tracking algorithm for nonlinear systems. We extend the class of applicable systems beyond the ones considered in \cite{herbert2017fastrack, Singh2018RobustTW}. 
Our  contributions are: 
\begin{enumerate}
    \item We introduce an approach for parametrizing the state and input constraints in the MPC planner with  parameters $\theta$. Using SOS programming \cite{Parrilo:00, Jarvis:05} \emph{offline}, we synthesize a parametric error bound $\mathcal{O}^\theta$ for the containment of error between planner and tracker states, and an associated feedback control policy $\kappa^\theta$. 
    
    \item We then solve an optimization problem \emph{offline} to pick the optimal parameter $\theta^\star$, that gives the ``widest" planner state constraint set $\hat{\mathcal{X}}^{\theta^\star}$ 
    such that, when enlarged by the error bound $\mathcal{O}^{\theta^\star}$, it is contained in the constraint set $\mathcal{X}$. Contrary to approaches such as  \cite{Singh2018RobustTW, YinStan2019}, this provides a systematic and optimal way of designing the planner and the associated error bound. 
    
    \item We solve an MPC  problem for the planner imposing constraints $\hat{\mathcal{X}}^{\theta^\star}$ and $\hat{\mathcal{U}}^{\theta^\star}$, and use input policy $\kappa^{\theta^\star}$ to control the plant. If the planner MPC problem is feasible, then satisfaction of all safety constraints are guaranteed for the plant. We demonstrate this with a detailed numerical example. 
\end{enumerate}

\subsection{Notation}
For $\xi \in \mathbb{R}^n$, $\mathbb{R}[\xi]$ represents the set of polynomials in $\xi$ with real coefficients, and $\R^{m}[\xi]$ and $\R^{m\times p}[\xi]$ denote all vector and
matrix valued polynomial functions. The subset $\Sigma[\xi] := \{p = p_1^2 + p_2^2 + ... + p_M^2 : p_1, ..., p_M \in \mathbb{R}[\xi]\}$ of $\mathbb{R}[\xi]$ is the set of SOS polynomials in $\xi$. Row $k$ of a matrix $A$, and element $k$ of a vector $b$ are denoted by $(A)_k$ and $(b)_k$ respectively. Unless defined otherwise, notation $x^j$ denotes a variable $x$ used in the $j$'th iteration of an iterative algorithm. The symbol ``$\leq$'' represents component-wise inequality.

\section{Problem Setup}
In this paper, the control framework as shown in Fig.~\ref{fig:vehicles} has two layers: (i) planning layer (planner), where a planning trajectory with a long time-horizon is generated using a planning model and Model Predictive Control (MPC); (ii) tracking layer (tracker), where tracking control signals are computed for the true plant to track the planned trajectories with bounded error. 
\begin{figure}[t]
	\centering
	\includegraphics[width=0.45\textwidth]{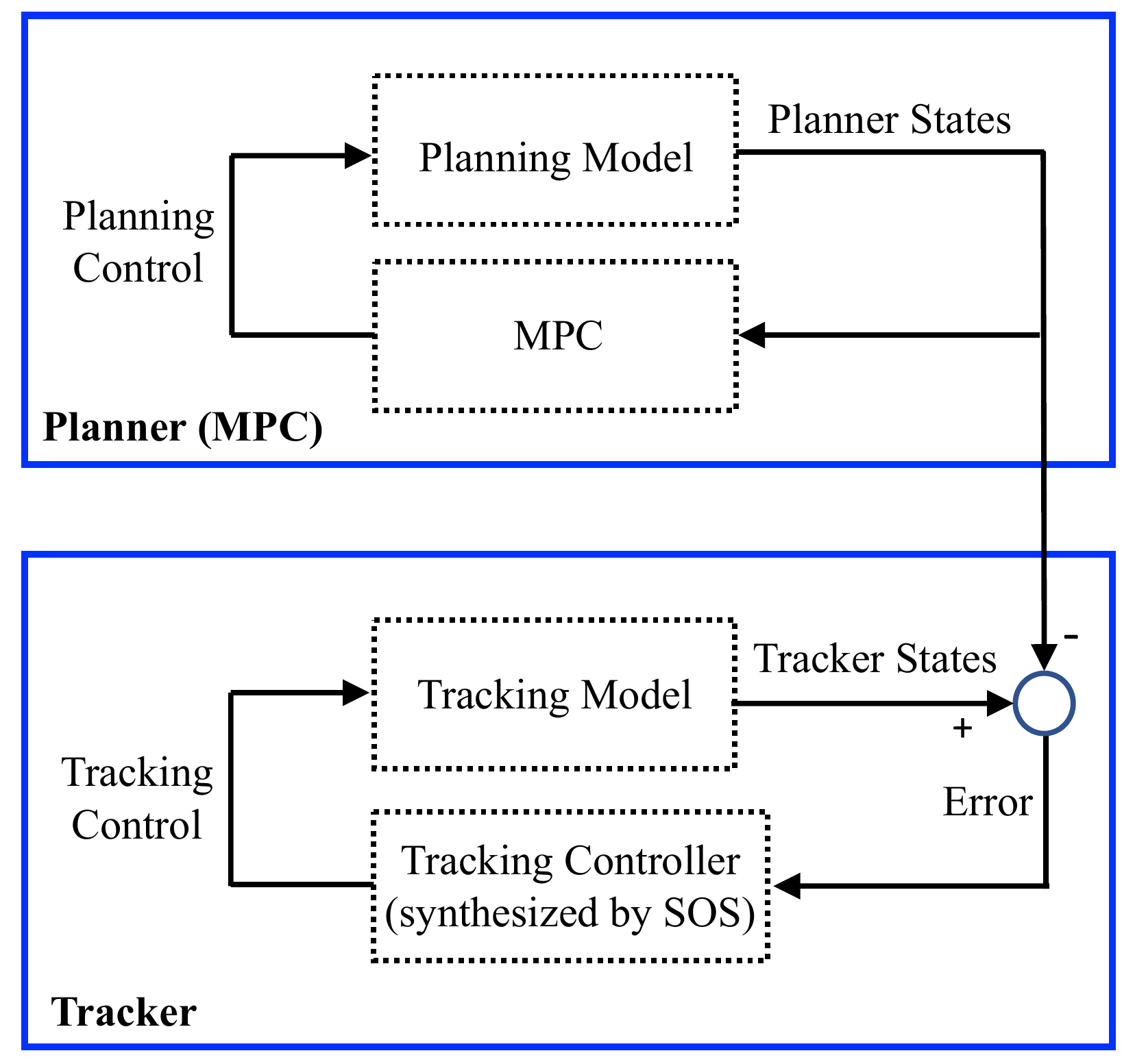}
	\caption{Control framework}
	\label{fig:vehicles}    
\end{figure}

\subsection{Tracking Model}
The high-fidelity model of the plant is referred to as the tracking model. This is an uncertain, input-affine, nonlinear system with parametric uncertainty $\delta(t)$, and is given as 
\begin{align}\label{eq:nonl_system}
    \dot{x}(t) = f(x(t), \delta(t)) + g(x(t), \delta(t))u(t),~\forall t \geq 0, 
\end{align}
where $x(t) \in \mathcal{X} \subseteq \R^n, u(t) \in \mathcal{U} \subseteq \R^m$, $\delta(t) \in \Delta \subseteq \R^{n_{\delta}}$, $f: \R^n \times \R^{n_\delta} \rightarrow \R^n$, and $g : \R^n \times \R^{n_\delta} \rightarrow \R^{n \times m}$. The sets $\mathcal{X}$ and $\mathcal{U}$ are state and control constraint sets imposed on the tracking model, and the set $\Delta := \{\delta \in \R^{n_\delta}: p_\delta(\delta) \leq 0\}$ defines the set of disturbances, where $p_\delta : \R^\delta \rightarrow \R$ is specified by the designer. 

\subsection{Planning Model}
The low-fidelity model, also referred to as the planning model, is a simplified (e.g. linearized) and potentially low-dimensional version of the tracking model, given by
\begin{align}\label{eq:low-fide-model}
    \dot{\hat{x}}(t) = \hat{f}(\hat{x}(t)) + \hat{g}(\hat{x}(t))\hat{u}(t),~\forall t \geq 0,
\end{align}
where $\hat{x}(t) \in \R^{\hat{n}}$, $\hat{u}(t) \in \R^{\hat{m}}$, $\hat{f} : \R^{\hat{n}} \rightarrow \R^{\hat{n}}$ and $\hat{g} : \R^{\hat{n}} \rightarrow \R^{\hat{n}} \times \R^{\hat{m}}$. For now we assume the high- and low-fidelity models have the same state dimension $\hat{n} = n$. The notation $\hat{n}$ is retained here for use in Section \ref{ModelReduction}, where the case when $\hat{n} \leq n$ is addressed.

\subsection{Error Dynamics}
Accounting for the the difference between the states of the low and high-fidelity models yields the error-states $e(t) = x(t) - \hat{x}(t)$. The error dynamics are given as, for all $t \ge 0$,
\begin{align}\label{eq:err_dyn}
& \dot{e}(t)= f_e(e(t),\hat{x}(t),\hat{u}(t), \delta(t)) + g_e(e(t),\hat{x}(t), \delta(t))u(t),
\end{align}
where $f_e(e,\hat{x},\hat{u},\delta) := f(e + \hat{x}, \delta) - \hat{f}(\hat{x}) - \hat{g}(\hat{x})\hat{u}$,  and  $g_e(e,\hat{x},\delta) := g(e + \hat{x}, \delta)$. Let $\mathcal{K}_{\mathcal{U}}:= \{\kappa: \R^n \times \R^{\hat{n}} \times \R^{\hat{m}} \times \R^{n_\delta} \rightarrow \mathcal{U}\}$ define a set of admissible error-state feedback control law. Notice that \eqref{eq:err_dyn} allows for dependence on $\hat{x}$ which is an extension to a richer class of systems than \cite{Singh2018RobustTW, herbert2017fastrack}.

\begin{assumption}\label{ass:initial}
Assume the initial condition of error-state, $e(0)$, starts within the set $\Omega := \{e \in \R^n: p_e(e) \leq 0\}$, that is, $e(0) \in \Omega$, where $p_e : \R^n \rightarrow \R$ is specified by the designer. 
\end{assumption}

\subsection{Planner Formulation}
In the planner, the MPC that generates online planning trajectories solves
\begin{equation} \label{eq:MPC_init}
    \begin{array}{llll}
        \displaystyle \min_{\hat{U}_t} & \sum \limits_{k=0}^{N-1}(\hat{x}_{k|t}^\top Q \hat{x}_{k|t} + \hat{u}_{k|t} R \hat{u}_{k|t}  ) + \hat{x}_{N|t}^\top P_N \hat{x}_{N|t} \vspace{2mm} \\ 
        \ \  \text{s.t. }  & 
         \hat{x}_{k+1|t} = \hat{F}_d(\hat{x}_{k|t}, \hat{u}_{k|t}, T_s),\\
         & \hat{x}_{k|t} \in \hat{\mathcal{X}},~ \hat{u}_{k|t} \in \hat{\mathcal{U}},\\
         & \forall k \in \{0,...,N-1\},\\
         & \hat{x}_{t|t} = \hat{x}(t),~\hat{x}_{N|t} \in \hat{\mathcal{X}}_N \subseteq \hat{\mathcal{X}},
    \end{array}
\end{equation}
with $Q,R,P_N \succ 0$, where $\hat{F}_d$ is system \eqref{eq:low-fide-model} discretized with sampling time $T_s$. Let $\hat{x}_{k|t}$ be the predicted planner states at time $t$ with predicted planner inputs $\hat{U}_t = [\hat{u}_{0|t}, \hat{u}_{1|t}, \dots, \hat{u}_{k-1|t}] \in \R^{\hat{m} \times k} $ for all $k \in  \{1,...,N\}$. 
Each prediction instant $k \in \{1,...,N\}$ represents look-ahead time of $kT_s$. 
The planner constraint sets are defined as
\begin{subequations}\label{eq:nomCon}
\begin{align}
    \hat{\mathcal{X}}:= & \{\hat{x} \in \R^{\hat{n}}: \hat{p}_x (\hat{x}) \leq \hat{h}_x \} ,\\
    \hat{\mathcal{U}}:= & \{\hat{u} \in \R^{\hat{m}}: \hat{p}_u (\hat{u}) \leq \hat{h}_u \},
\end{align}
\end{subequations}
with $\hat{p}_x : \R^{\hat{n}} \rightarrow \R, \ \hat{p}_u : \R^{\hat{m}} \rightarrow \R, \ \hat{h}_x \in \R, \ \hat{h}_u \in \R$ chosen by the designer. Terminal conditions $\hat{\mathcal{X}}_N$ and $P_N$ are chosen to ensure feasibility and stability properties \cite{borrelli2017predictive}. After solving \eqref{eq:MPC_init} at any time $t$, we apply the first optimal input $\hat{u}^\star(t) = \hat{u}^\star_{0|t}$ \emph{only} to low-fidelity planner system \eqref{eq:low-fide-model}. We then re-solve \eqref{eq:MPC_init} at next time instant $t+T_s$.

\subsection{Tracker Formulation}
\begin{definition}\label{MPIS}{Robust Infinite-Time Forward Reachable Set:}
Consider the closed-loop error dynamics obtained from \eqref{eq:err_dyn} for all $t \ge 0$, under a given control law $\kappa \in \mathcal{K}_{\mathcal{U}}$ as 
\begin{align} \label{eq:err_closedloop}
    &\dot{e}(t) = f_e(e(t),\hat{x}(t),\hat{u}(t),\delta(t)) + \nonumber \\
    & ~~~~~~~~~g_e(e(t),\hat{x}(t),\delta(t))\kappa(e(t), \hat{x}(t),\hat{u}(t),\delta(t)), 
\end{align}
with $\hat{x}(t)$ and $\hat{u}(t)$ constrained by \eqref{eq:nomCon} for all times $t\geq 0$. Then a \emph{robust infinite-time forward reachable set} $\mathcal{O}$ of $\Omega$, for a given feedback $\kappa$, is defined as 
\begin{align*}
    &\mathcal{O} := \{e(t) \in \R^n: \exists e(0) \in \Omega, \ \hat{x}: \R_+ \rightarrow \hat{\mathcal{X}},~ \hat{u}: \R_+ \rightarrow \hat{\mathcal{U}}, \\
    &~~~~~~~~~~~\delta: \R_+ \rightarrow \Delta, ~  t \ge 0, ~ \text{s.t.} ~ e(t) ~ \text{is a solution to}~ \eqref{eq:err_closedloop}\}.
\end{align*}
\end{definition}
We assume that $\mathcal{O}$ is a compact set. 
The tracking control synthesizes a error-state feedback policy $u(t) = \kappa(e(t),\hat{x}(t),\hat{u}(t),\delta(t))$ with $\kappa \in \mathcal{K}_{\mathcal{U}}$ ensuring containment of the error-states within such an $\mathcal{O}$. We refer to that $\mathcal{O}$ as an ``error bound'', and $\kappa$ as the corresponding ``tracking control" law and they can be obtained by following \cite{YinStan2019} using Sum-of-Squares (SOS) programming. \textcolor{black}{\textcolor{black}{$\mathcal{O}$ is a function of $\hat{\mathcal{X}}, \hat{\mathcal{U}}$ and $\Delta$.} As the volumes of $\hat{\mathcal{X}}$, $\hat{\mathcal{U}}$ and $\Delta$ increase, we tend to get a larger error bound $\mathcal{O}$.}
\begin{remark}
Note that since the planner MPC problem \eqref{eq:MPC_init} is not posed in continuous time, the guarantees of feasibility of planner constraints \eqref{eq:nomCon} hold only at sampled time instants, assuming perfect discretization (although this is valid for linear systems, for nonlinear systems variational methods can be used for obtaining arbitrarily low discretization errors \cite{nair2018discrete}).  
We hereby assume that the planner sample frequency is chosen high enough that all continuous time guarantees hold for this planner-tracker synthesis work. 
\end{remark}

\section{Parametric Approach to Planner--Tracker Design }
The primary goal is to ensure constraint satisfaction on the state $x(t)$ of the tracker evolving according to \eqref{eq:nonl_system} under the control law $\kappa$, i.e $x(t) \in \mathcal{X}$ for all $t \geq 0$. For this we must make sure
\begin{align}
    \hat{\mathcal{X}} \oplus \mathcal{O} \subseteq \mathcal{X}\label{eq:feasCon}.
\end{align}
Note that if $\hat{\mathcal{X}}$ is chosen to be of small volume, \textcolor{black}{the corresponding $\mathcal{O}$ is small, and it is very likely that \eqref{eq:feasCon} will hold}, but it might leave too small room for \eqref{eq:MPC_init} to be feasible. If $\hat{\mathcal{X}}$ is chosen to be too large, \eqref{eq:feasCon} might be violated. To address this trade-off between planner permissiveness and tracker safety, we propose a parametric approach, where we parametrize planner constraint sets \eqref{eq:nomCon} as $\hat{\mathcal{X}}^\theta$ and $\hat{\mathcal{U}}^\theta$, $\theta \in \Theta$. The set $\Theta$ is defined as $\Theta := \{\theta \in \R^{n_\theta}: p_\theta(\theta) \leq 0\}$, where $p_\theta : \R^\theta \rightarrow \R$ is picked by the user. 
Correspondingly, we compute a \emph{parametric forward reachable set} $\mathcal{O}^\theta$ of $\Omega$, where
\begin{align*}
    &\mathcal{O}^\theta := \{e(t): \exists e(0) \in \Omega, \ \hat{x}: \R_+ \rightarrow \hat{\mathcal{X}}^\theta,~ \hat{u}: \R_+ \rightarrow \hat{\mathcal{U}}^\theta, \\
    & ~~~~~~~~~~ \delta: \R_+ \rightarrow \Delta, ~  t \ge 0, ~ \text{s.t.} ~ e(t) ~ \text{is a solution to}~ \eqref{eq:err_closedloop}\},
\end{align*}
and its associated parametric control law $\kappa^\theta$. $\mathcal{O}^\theta$ is referred to as a ``parametric error bound". The existence of a specific parameter $\tilde{\theta} \in \Theta$ which satisfies 
\begin{align}\label{eq:con_feasPar}
    \hat{\mathcal{X}}^{\tilde{\theta}} \oplus \mathcal{O}^{\tilde{\theta}} \subseteq \mathcal{X},
\end{align}
and for which \eqref{eq:MPC_init} is feasible, ensures safety: $x(t) \in \mathcal{X}$ for all $t \geq 0$. We parametrize constraint sets \eqref{eq:nomCon} using $\theta \in \Theta$ as
\begin{align}
    \hat{\mathcal{X}}^\theta:= & \{\hat{x} \in \mathbb{R}^{\hat{n}} : \hat{p}_x (\hat{x}) \leq \hat{h}^\theta_x \},  \label{eq:nomConParam1}\\
    \hat{\mathcal{U}}^\theta:= & \{\hat{u} \in \mathbb{R}^{\hat{m}}: \hat{p}_u (\hat{u}) \leq \hat{h}^\theta_u \} \label{eq:nomConParam2}.
\end{align}
where $\hat{h}_x^\theta: \R^{n_\theta} \rightarrow \R$ and $\hat{h}_u^\theta: \R^{n_\theta} \rightarrow \R$.
Given the parametrized constraint sets \eqref{eq:nomConParam1}--\eqref{eq:nomConParam2},
we take two steps, 
\begin{enumerate}[A)]
\item compute a parametric  error bound $\mathcal{O}^\theta$, and an associated feedback policy denoted by $\kappa^\theta$, which may vary as the parameter $\theta$ is varied;
\item solve an optimization problem to pick the ``best" $\tilde{\theta}$ that gives the most permissive planner (widest $\hat{\mathcal{X}}^{\tilde{\theta}}$) subject to the safety constraint \eqref{eq:con_feasPar}. 
\end{enumerate}
These steps are elaborated in the following sections. 

\subsection{Parametric Error Bound $\mathcal{O}^\theta$}
We use the following Theorem \ref{Theo1} to compute a parametric error bound $\mathcal{O}^\theta$, as well as an associated feedback control policy, denoted as $\kappa^\theta$. Note that we use the same symbol for a particular real variable in the algebraic
statements as well as the corresponding signal in the dynamical systems, after dropping the time-series argument.

Consider tracker input $u$ defined in \eqref{eq:nonl_system}. We assume the set of constraints on $u$ is a polytope $\mathcal{U} = \{u \in \R^m : H u \leq h\}$, where $H \in \R^{N_0 \times m}$, $h \in \R^{N_0}$, and we overload the notation $\mathcal{K}_\mathcal{U}$ as $\mathcal{K}_{\mathcal{U}}:= \{\kappa: \R^n \times \R^{\hat{n}} \times \R^{\hat{m}} \times \R^{n_\delta} \times \R^{n_\theta} \rightarrow \mathcal{U}\}$.

\begin{theorem} \label{Theo1}
Let Assumption~\ref{ass:initial} hold. Given the error dynamics with mappings $f_e: \mathbb{R}^n \times \mathbb{R}^{\hat{n}} \times \mathbb{R}^{\hat{m}} \times \R^{n_\delta} \rightarrow \mathbb{R}^n$, $g_e: \mathbb{R}^n \times \mathbb{R}^{\hat{n}} \times \R^{n_\delta} \rightarrow \mathbb{R}^n$, $\gamma \in \R$, $\hat{\mathcal{X}}^\theta \subseteq \mathbb{R}^{\hat{n}}$, $\hat{\mathcal{U}}^\theta \subseteq \mathbb{R}^{\hat{m}}$, $\Theta \subseteq \R^{n_\theta}$, $\Delta \subseteq \R^{n_\delta}$, $\Omega \subseteq \R^n$, $H \in \R^{N_0 \times m}$ and $h \in \R^{N_0}$, if there exists a $\mathcal{C}^1$ function $V: \mathbb{R}^{n} \times \mathbb{R}^{n_\theta} \rightarrow \mathbb{R}$, and $\kappa: \mathbb{R}^n \times \mathbb{R}^{\hat{n}} \times \mathbb{R}^{\hat{m}} \times \R^{n_\delta} \times \R^{n_\theta} \rightarrow \mathbb{R}^m$, such that for all $\delta \in \Delta, \hat{x} \in \hat{\mathcal{X}}^\theta,  \hat{u} \in \hat{\mathcal{U}}^\theta,$ the following constraints hold,
\begin{subequations}
\begin{align}
& \frac{\partial V(e,\theta)}{\partial e}\cdot(f_e(e, \hat{x}, \hat{u}, \delta)+g_e(e, \hat{x}, \delta)\kappa(e,\hat{x},\hat{u},\delta,\theta)) \leq 0, \nonumber \\
&  ~~~~~~~~~~~~~~~~~~~\forall (e, \theta) \in \R^n \times \Theta, \ \text{s.t.} \ V(e, \theta) = \gamma, \label{eq:Vcond} \\
& \{e: V(e,\theta) \leq \gamma\} \subseteq \{e : H \kappa(e, \hat{x}, \hat{u}, \delta, \theta) \leq h\}, \ \forall \theta \in \Theta, \nonumber \\
\label{eq:ControlConstr}\\
 & \textcolor{black}{\Omega \times \Theta \subseteq \{(e,\theta): V(e,\theta) \leq \gamma\},} \label{eq:InitCond} 
\end{align}
\end{subequations}
then the $\theta$-dependent sub-level set \begin{align*}
\mathcal{O}^\theta := \{e: V(e,\theta) \leq \gamma\}
\end{align*}
is a parametric forward reachable set of $\Omega$ under the control policy $\kappa^\theta:= \kappa(\cdot,\cdot,\cdot,\cdot,\theta) \in \mathcal{K}_{\mathcal{U}}$.
\end{theorem}
\begin{proof}
We have, $\theta$ as a vector of uncertain parameters with dynamics $\dot{\theta} = 0$. Let $\theta(0) = \theta^0$, for all possible $\theta^0 \in \Theta$. For all the augmented states $(e(0),\theta(0)) \in \Omega \times \Theta \subseteq \{(e,\theta): V(e,\theta) \leq \gamma\}$, (i.e. $e(0) \in \{e: V(e,\theta^0) \leq \gamma\}$), we have $e(t) \in \{e: V(e,\theta^0)\leq \gamma \}$, implying $\{e: V(e,\theta) \leq \gamma\}$ is a parametric forward reachable set of $\Omega$.
\end{proof}

\begin{remark}
Since $\mathcal{O}^\theta$ is also a positive invariant set for error-states for all $\theta \in \Theta$, after we obtain it based on $\Omega$, it can then serve as the set of initial conditions for error-states. 
\end{remark}

We use Sum-of-Squares (SOS) programming \cite{Parrilo:00, Jarvis:05} in finding storage function $V$ and control law $\kappa^\theta$ by solving the following non-convex optimization problem. We restrict $p_e \in \R[e]$, $p_\theta \in \R[\theta]$, $p_\delta \in \R[\delta]$, $\hat{p}_x \in \R[\hat{x}]$, $\hat{p}_u \in \R[\hat{u}]$, $f_e \in \R^{n}[(e, \hat{x}, \hat{u}, \delta)]$, $g_e \in \R^{n \times m}[(e, \hat{x}, \delta)]$, $V \in \R[(e, \theta)]$ and $\kappa \in \R[(e,\hat{x},\hat{u},\delta,\theta)]$. 
\begingroup
\allowdisplaybreaks
\begin{subequations} \label{eq:Vk_sosopt}
    \begin{align}
        \displaystyle \min_{V, \kappa, s} ~&
        \mathrm{volume}(\mathcal{O}^\theta) \vspace{2mm} \nonumber \\ 
        \ \  \mathrm{s.t.} ~
        & s_2 \in \R[(e,\hat{x},\hat{u},\delta,\theta)],\  s_{13}, s_{14} \in \Sigma[(e,\theta)], \nonumber \\
        & s_j \in \Sigma[(e,\hat{x},\hat{u},\delta,\theta)], ~~~~\hspace{0.3mm}\forall j \in \{3,\dots,6\}, \nonumber \\
        & (s_{l})_k \in \Sigma[(e,\hat{x},\hat{u},\delta,\theta)], \ \forall l \in \{8,\dots,12\}, \label{eq:Vk_sosopt1} \\
        & -\left(\frac{\partial V}{\partial \theta}\right)_i + (s_{7})_i \cdot p_\theta \in \Sigma[(e,\theta)],  (s_{7})_i \in \Sigma[(e,\theta)] \nonumber  \\
        &~~~~~~~~~~~~~~~~~~~~~~~~~~~~~~~~~\forall i \in \{1,\dots,n_\theta\},  \label{eq:Vk_sosopt3}\\
         &-\frac{\partial V}{\partial e}\cdot (f_e + g_e \kappa) - s_2\cdot (V - \gamma)  + s_3 \cdot (\hat{p}_x - \hat{h}_x^\theta) \nonumber \\
         & \quad + s_4 \cdot (\hat{p}_u - \hat{h}_u^\theta) + s_5 \cdot p_\theta  + s_6 \cdot p_\delta \nonumber \\
         & \quad ~~~~~~ \in \Sigma[(e, \hat{x}, \hat{u}, \delta, \theta)],  \label{eq:Vk_sosopt2}\\
         & (h)_k - (H)_k \kappa + (s_{8})_k \cdot (V - \gamma) + (s_{9})_k \cdot (\hat{p}_x - \hat{h}_x^\theta)  \nonumber \\
         & \quad + (s_{10})_k \cdot (\hat{p}_u - \hat{h}_u^\theta) + (s_{11})_k \cdot p_\theta + (s_{12})_k \cdot p_\delta \nonumber \\
         & \quad ~~~~~~ \in \Sigma[(e, \hat{x}, \hat{u}, \delta, \theta)], ~\forall k \in \{1,...,N_0\}, \label{eq:Vk_sosopt4} \\
         &-(V - \gamma) + s_{13} \cdot p_e + s_{14} \cdot p_{\theta} \in \Sigma[(e,\theta)], \label{eq:Vk_sosopt5}
    \end{align}
\end{subequations}
\endgroup
SOS polynomials $s$ serve as the S-procedure certificates, and are usually referred to as ``multiplier polynomials''. Constraints \eqref{eq:Vk_sosopt2}--\eqref{eq:Vk_sosopt5}, when feasible, are sufficient conditions for \eqref{eq:Vcond}--\eqref{eq:InitCond}, respectively. 
The rationale for constraint \eqref{eq:Vk_sosopt3} is elaborated in Proposition~\ref{prop1}. Solving optimization \eqref{eq:Vk_sosopt} directly can be challenging, since it is bi-linear in decision variables $V$ and $(\kappa, s_2~ \mathrm{and~} (s_{8})_k)$. 
Similar to \cite{YinBackward19}, in Algorithm~\ref{alg:alg1} we decompose optimization \eqref{eq:Vk_sosopt} into two convex sub-problems to iteratively search between two sets of decision variables. Note that the initialization $V^0$ to Algorithm \ref{alg:alg1} can be computed using \cite[Algorithm~2]{YinStan2019}. 


\begin{algorithm} [h]
	\caption{Computing $\mathcal{O}^\theta$ and $\kappa^\theta$}
	\label{alg:alg1}
	\begin{algorithmic}[1]
	\Require{function $V^0$ such that \eqref{eq:Vk_sosopt1}--\eqref{eq:Vk_sosopt5} are feasible by proper choice of $s, \kappa, \gamma$ and \textcolor{black}{sub-level sets of $V^0$ are bounded. Maximum iteration count $N_\mathrm{iter}$.}}
	\Ensure{($\kappa$, $\gamma$, $V$) such that with the volume of $\mathcal{O}^\theta$ having been shrunk.}
	\For{$j = 1:N_\mathrm{iter}$}
		\State $\boldsymbol{\gamma}$\textbf{-step}: decision variables: $(s, \kappa,\gamma)$.
			
		Minimize $\gamma$ subject to \eqref{eq:Vk_sosopt1}, \eqref{eq:Vk_sosopt2}--\eqref{eq:Vk_sosopt5} using 
		
		$V = V^{j-1}$. This yields $(s_2^j, (s_{8})_k^j, \kappa^j)$, for all 
		
		$k \in \{1, ..., N_0\}$ and optimal cost $\gamma^j$.
		
		\State $\boldsymbol{V}\textbf{-step}$: decision variables: $V$ and all the multiplier 
		
		polynomials $s$ except $(s_2~\mathrm{and}~ (s_{8})_k)$. Maximize the 
		
		feasibility \cite{YinBackward19} subject to \eqref{eq:Vk_sosopt1}--\eqref{eq:Vk_sosopt5} as well as 
		
		$s_0, s_1 \in \Sigma[(e,\theta)],$ and 
		\begin{align}
		-s_0 \cdot (V^{j-1}& - \gamma^j) +(V - \gamma^j) \nonumber \\
		& \quad ~~~~~~ \ + s_1 \cdot p_\theta \in \Sigma[(e, \theta)], \label{eq:DescentCondi}
		\end{align}
		
		using ($\gamma = \gamma^j$, $s_2 = s_2^j$, $(s_{8})_k = (s_{8})_k^j$, $\kappa=\kappa^j$), 
		
		for all $k \in \{1, ..., N_0\}$. This yields $V^j$.
		\EndFor
		\end{algorithmic}
\end{algorithm}
The constraint \eqref{eq:DescentCondi} enforces the sub-level set certified by the $V$-step, $\{(e,\theta): V^j(e,\theta) \leq \gamma^j\}$, to be contained by the sub-level set from the $\gamma$-step, $\{(e,\theta): V^{j-1}(e,\theta) \leq \gamma^j\}$,  for all $\theta \in \Theta$.


\subsection{Optimal Parameter Selection}
Next, we need to pick the ``optimal" $\tilde{\theta}$, denoted by $\theta^\star$, which gives the ``widest" $\hat{\mathcal{X}}^{\tilde{\theta}}$ subject to \eqref{eq:con_feasPar}. 
The Minkowski sum of $\hat{\mathcal{X}}^{\theta}$ and $\mathcal{O}^{\theta}$ in \eqref{eq:con_feasPar} can be expressed as follows:
\begin{align}
    & \hat{\mathcal{X}}^\theta \oplus \mathcal{O}^\theta \nonumber \\
    =& \{x \in \mathbb{R}^n : x = \hat{x} + e, \ \hat{p}_x (\hat{x}) \leq \hat{h}_x^\theta, \ V(e, \theta) \leq \gamma\}, \nonumber \\
    =& \{x \in \mathbb{R}^n : \hat{p}_x (\hat{x}) \leq \hat{h}_x^\theta, \ V(x - \hat{x}, \theta) \leq \gamma \}, \nonumber \\
    \text{or} \ =& \{x \in \mathbb{R}^n : \hat{p}_x (x-e) \leq \hat{h}_x^\theta, \ V(e, \theta) \leq \gamma \}. \label{eq:Mink_sum}
\end{align}
We assume $\mathcal{X}$ is a semi-algebraic set, which is a sub-level set of a given polynomial function $p(\cdot)$. That is $\mathcal{X} = \{x \in \mathbb{R}^n: p(x) \ge 0\}$. 

\subsubsection*{Optimal Parameter Selection by Sum-of-Squares}
Replacing the constraint \eqref{eq:con_feasPar} with the reformulation as in \eqref{eq:Mink_sum}, we pose the following optimization problem by applying the polynomial S-procedure to \eqref{eq:Mink_sum} to obtain $\theta^\star$. Assume that $\hat{h}_x^\theta$ and $\hat{h}_u^\theta$ are chosen in a way that when $\theta$ grows, $\hat{h}_x^\theta$ and $\hat{h}_u^\theta$ grow as well, by making sure
\begin{align}
\frac{\partial \hat{h}_x^\theta}{\partial \theta} \ge 0 \ \mathrm{and} \ \frac{\partial \hat{h}_u^\theta}{\partial \theta} \ge 0, \ \forall \theta \in \Theta. \label{eq:hx_choice}
\end{align}
Therefore, to find the most permissive constraint sets for the MPC in \eqref{eq:MPC_init}, the summation of all the elements of $\theta$ is chosen as the reward function in
\begin{equation} \label{eq:theta_sosopt}
    \begin{array}{llll}
        \displaystyle \max_{\theta, s_a, s_b} & \sum_{i=1}^{n_{\theta}}(\theta)_i \vspace{2mm} \\ 
        \ \  \text{s.t. }  & \theta \in \Theta, \ s_a, s_b \in \Sigma[(x,e)], \\ 
        & p + s_a \cdot (\hat{p}_x (x - e) - \hat{h}_x^{\theta}) \\
        & \quad + s_b \cdot (V(e, \theta) - \gamma) \in \Sigma[(x,e)],       
    \end{array}
\end{equation}
where $s_a$ and $s_b$ are polynomial multipliers. 
However, \eqref{eq:theta_sosopt} is bi-linear in two sets of decision variables: multipliers ($s_a, s_b$) and ($V(e, \theta))$, $\hat{h}_x^{\theta}$) that are nonlinear functions in $\theta$. Although \eqref{eq:theta_sosopt} is convex in ($s_a$, $s_b$) when $\theta$ is fixed, \eqref{eq:theta_sosopt} is not necessarily convex in $\theta$ when fixing ($s_a$, $s_b$). 
We resolve this issue with the following Proposition. 

\begin{proposition} \label{prop1}
Imposing constraints \eqref{eq:Vk_sosopt3} on $V$, that is, $\frac{\partial V}{\partial \theta} \leq 0$, for all $(e,\theta) \in \R^n \times \Theta$, ensures
\begin{align} 
    \mathcal{O}^{\theta^a} \subseteq \mathcal{O}^{\theta^b}, \ \forall \theta^a \leq \theta^b,  \label{eq:O_grow}
\end{align}
where $\theta^a, \theta^b \in \Theta$. 
\end{proposition}
\begin{proof}
As $\frac{\partial V}{\partial \theta} \leq 0$,  for all $(e, \theta) \in \R^n \times \Theta$, we have $V(e, \theta^a) \ge V(e, \theta^b)$. If an error-state $e$ satisfies $V(e,\theta^a) \leq \gamma$, then it also satisfies $V(e,\theta^b) \leq \gamma$. 
\end{proof}

\subsubsection*{Reformulation for Iterative Convex Optimization}
We can iteratively solve \eqref{eq:theta_sosopt} with Linear Matrix Inequalities (LMIs) if we do not make $\theta$ a decision variable, and instead look for its maximum allowable box bound $\overline{\theta}$ such that $\theta \in [0,\bar{\theta}]$.
Thus, we solve the following reformulated SOS optimization problem as a tractable relaxation to \eqref{eq:theta_sosopt}:
\begin{equation} \label{eq:theta_ub_sos}
    \begin{array}{llll}
        \displaystyle \max_{\bar{\theta}, s_a, s_b, s_c} & \sum_{i=1}^{n_{\theta}}(\bar{\theta})_i \vspace{2mm} \\ 
        \ \  \text{s.t.}  & \bar{\theta} \in \Theta,\\
        & s_{a}, s_{b}, (s_{c})_i \in \Sigma[(x,e,\theta)],~\forall i \in \{1,\dots,n_\theta\}, \\ 
        & p + s_{a} \cdot (\hat{p}_x (x - e) - \hat{h}_x^{\theta}) + s_{b} \cdot (V(e, \theta) - \gamma)\\
        & ~ - \sum_{i=1}^{n_\theta} (s_{c})_i \cdot (\theta)_i \Big ((\bar{\theta})_i - (\theta)_i \Big)\in \Sigma[(x,e,\theta)]. \vspace{1mm}
    \end{array}
\end{equation}
When feasible, \eqref{eq:theta_ub_sos} is a sufficient condition for 
\begin{align}\label{eq:all_theta}
    \hat{\mathcal{X}}^\theta \oplus \mathcal{O}^\theta \subseteq \mathcal{X}, \ \forall \theta \in [0, \bar{\theta}].
\end{align}
Most importantly, optimization problem \eqref{eq:theta_ub_sos} is only bi-linear in $(s_{c})_i$ and $(\bar{\theta})_i$, and can be solved by iteratively searching between $(s_{c})_i$ and $(\bar{\theta})_i$ using Algorithm \ref{alg:alg2}. 
 
 \begin{algorithm} [H]
	\caption{Optimal $\theta$ Selection}
	\label{alg:alg2}
	\begin{algorithmic}[1]
	\Require{$\bar{\theta}^0$ such that constraints in \eqref{eq:theta_ub_sos} are feasible by proper choice of $s_a, s_b, (s_{c})_i$, for all $i \in \{1, ..., n_\theta\}$.}
	\Ensure{$\bar{\theta}$ that has been maximized}
	\For{$j = 1:N_\mathrm{iter}$}
	    \State $\boldsymbol{s}\textbf{-step}$: decision variables: $(s_a, s_b, (s_{c})_i)$.
		
		Maximize the feasibility subject to the constraints 
		
		in \eqref{eq:theta_ub_sos}, using $\bar{\theta} = \bar{\theta}^{j-1}$. This yields $(s_{c})_i^j$.
		
		\State $\boldsymbol{\bar{\theta}}$\textbf{-step}: decision variables: $(s_a, s_b, \bar{\theta})$.
			
		Maximize $\bar{\theta}$ subject to the constraints in \eqref{eq:theta_ub_sos}  using 
		
		$(s_{c})_i = (s_{c})_i^j$ for all $i \in \{1, ..., n_\theta\}$. This yields
		
		an optima $\bar{\theta}^j$ to \eqref{eq:theta_ub_sos}.
		
		\EndFor
		\end{algorithmic}
\end{algorithm}

\begin{assumption} \label{ass1}
We assume $\Theta$ is a box constraint set, and without loss of generality we choose $0$ as the lower bound.
\end{assumption}
\begin{proposition}
Assume Proposition~\ref{prop1} and Assumption~\ref{ass1} hold. Assume $\theta^\star$ and $\bar{\theta}^\star$ as the global optima to \eqref{eq:theta_sosopt} and \eqref{eq:theta_ub_sos} respectively. Then, the optimization problems \eqref{eq:theta_sosopt} and \eqref{eq:theta_ub_sos} are equivalent. That is, $\theta^\star = \bar{\theta}^\star$. 
\end{proposition}
\proof
 For all $\tilde{\theta} \leq \theta^\star$, by \eqref{eq:hx_choice}, it yields $\hat{\mathcal{X}}^{\tilde{\theta}} \subseteq \hat{\mathcal{X}}^{\theta^\star}$. It follows from \eqref{eq:O_grow} that $\mathcal{O}^{\tilde{\theta}} \subseteq \mathcal{O}^{\theta^\star}$. Since $\theta^\star$ satisfies $\hat{\mathcal{X}}^{\theta^\star} \oplus \mathcal{O}^{\theta^\star} \subseteq \mathcal{X}$, we have $\hat{\mathcal{X}}^{\tilde{\theta}} \oplus \mathcal{O}^{\tilde{\theta}} \subseteq \mathcal{X}$, for all $\tilde{\theta} \leq \theta^\star$. It follows from Assumption~\ref{ass1} that the feasible set of $\theta$ for \eqref{eq:theta_sosopt} is $\mathcal{F}_\theta := \{\theta \in \Theta: 0 \leq \theta \leq \theta^\star\}$, which from \eqref{eq:all_theta} implies $\theta^\star = \bar{\theta}^\star$. This proves the Proposition.
\endproof


\section{Model reduction} \label{ModelReduction}
In practice it may be desirable to simplify the low-fidelity model further by reducing the state dimension. To make the reduced states comparable to the original states, we define an appropriate map $\pi : \mathbb{R}^{\hat{n}} \rightarrow \mathbb{R}^{n}$, $\hat{n} \leq n$, and redefine the error-states as $e(t) = x(t) - \pi(\hat{x}(t))$. The map $\pi(\cdot)$ needs to be chosen with care according to the specific application and control objective. Accordingly, for the newly defined error-states, $f_e$ and $g_e$ in \eqref{eq:err_dyn} become
\begin{align}
    &f_e(e,\hat{x},\hat{u},\delta) := f(e + \pi(\hat{x}),\delta)  - \frac{\partial \pi}{\partial \hat{x}}(\hat{f}(\hat{x}) + \hat{g}(\hat{x})\hat{u}), \nonumber \\
    &g_e(e, \hat{x}, \delta) := g(e+\pi(\hat{x}), \delta).  \nonumber
\end{align}
Without any modification, optimization \eqref{eq:Vk_sosopt} can still be used to compute parametric error bounds and control law for the error dynamics with model reduction. However, the optimization for finding optimal parameter will need to change, since the constraint  \eqref{eq:con_feasPar} now becomes 
$\pi\left(\hat{\mathcal{X}}^\theta \right) \oplus \mathcal{O}^\theta \subseteq \mathcal{X}$, where $\pi\left(\hat{\mathcal{X}}^\theta \right) := \{\eta \in \mathbb{R}^n : \eta = \pi(\hat{x}), \ \hat{p}_x (\hat{x}) \leq \hat{h}_x^\theta\}.$ Then, the Minkowski sum of $\pi\left(\hat{\mathcal{X}}^\theta \right)$ and $\mathcal{O}^\theta$
amounts to 
\begin{align}
   & \pi\left(\hat{\mathcal{X}}^\theta \right) \oplus \mathcal{O}^\theta = \{x \in \mathbb{R}^n : x = \eta + e, \ \eta = \pi(\hat{x}), \nonumber \\
    &~~~~~~~~~~~~~~~~~~~~~~~~~~~  \hat{p}_x(\hat{x}) \leq \hat{h}_x^\theta, \ V(e, \theta) \leq \gamma\}, \nonumber \\
    & = \{x \in \mathbb{R}^n: \hat{p}_x (\hat{x}) \leq \hat{h}_x^\theta, \ V(x - \pi(\hat{x}), \theta) \leq \gamma\}. \nonumber 
\end{align}
To render the parameter selection process tractable, we look for a maximum allowable box bound $\bar{\theta}$ that makes $\pi\left(\hat{\mathcal{X}}^\theta \right) \oplus \mathcal{O}^\theta \subseteq \mathcal{X}, \ \forall \theta \in [0, \bar{\theta}]$ feasible by replacing the constraint in \eqref{eq:theta_ub_sos} with the following constraint 
\begin{align}
    &p + s_d \cdot \left(\hat{p}_x (\hat{x}) - \hat{h}_x^{\theta}\right) + s_e \cdot \left(V (x - \pi(\hat{x}), \theta) - \gamma \right) \nonumber \\ 
    &~~~~~~- \sum_{j=1}^{n_\theta}(s_{f})_j \cdot (\theta)_j\left((\bar{\theta})_j - (\theta)_j\right)\in \Sigma[(x,\hat{x},\theta)], \nonumber
\end{align}
where $s_d, s_e, (s_{f})_j \in \Sigma[(x, \hat{x}, \theta)]$ for all $j \in \{1,...,n_\theta\}$.

\section{Planner Feasibility and Tracker Constraint Satisfaction}
Once $\bar{\theta}^\star$ is fixed, the high level planner, that is the MPC just has to solve the following reformulation of \eqref{eq:MPC_init}:
\begin{equation} \label{eq:MPC_plan}
    \begin{array}{llll}
        \displaystyle \min_{\hat{U}_t} & \sum \limits_{k=0}^{N-1}(\hat{x}_{k|t}^\top Q \hat{x}_{k|t} + \hat{u}_{k|t} R \hat{u}_{k|t}  ) + \hat{x}_{N|t}^\top P_N \hat{x}_{N|t} \vspace{2mm} \\ 
        \ \  \text{s.t. }  & 
         \hat{x}_{k+1|t} = \hat{F}_d(\hat{x}_{k|t}, \hat{u}_{k|t}, T_s),\\
         & \hat{x}_{k|t} \in \hat{\mathcal{X}}^{\bar{\theta}^\star},~ \hat{u}_{k|t} \in \hat{\mathcal{U}}^{\bar{\theta}^\star},\\
         & \forall k \in \{0,\dots,N-1\},\\
         & \hat{x}_{t|t} = \hat{x}(t),~\hat{x}_{N|t} \in \hat{\mathcal{X}}_N \subseteq \hat{\mathcal{X}}^{\bar{\theta}^\star},
    \end{array}
\end{equation}
with $Q,R,P_N \succ 0$. We solve \eqref{eq:MPC_plan} at any time $t$ and then apply the first input 
\begin{align}\label{eq:mpc_cl}
    \hat{u}(t) = \hat{u}^\star_{0|t}
\end{align}
to \eqref{eq:low-fide-model}. We then re-solve \eqref{eq:MPC_plan} at the next instant $t+T_s$ and repeat in receding horizon fashion.  
\begin{assumption}\label{ass:mpc_feas}
We assume recursive feasibility of \eqref{eq:MPC_plan}. That is, if \eqref{eq:MPC_plan} is feasible at time $t = 0$, it remains feasible for all times $t \geq 0$, when \eqref{eq:mpc_cl} is applied to \eqref{eq:low-fide-model}. 
\end{assumption}

Recursive feasibility of  a nonlinear planner can be achieved by picking a ``long" prediction horizon $N$ as mentioned in \cite{michMayne, mayne2000constrained}. However in this case, problem \eqref{eq:MPC_plan} remains non-convex. An alternative way of ensuring recursive feasibility of \eqref{eq:MPC_plan} while solving a convex problem is by resorting to linear time invariant planner dynamics $\hat{x}(t+T_s) = \hat{A}\hat{x}(t) + \hat{B}\hat{u}(t)$ and then appropriately choosing terminal conditions $\hat{\mathcal{X}}_N$ and $P_N$. Matrices $\hat{A}, \hat{B}$ can be chosen with OLS approximation \cite{dean2017sample} of \eqref{eq:low-fide-model}. 

\begin{proposition}
Let problem \eqref{eq:theta_ub_sos} be feasible. Let Assumption~\ref{ass:mpc_feas} hold true. Assume initial error-states satisfy $e(0) \in \Omega$.  Then system variable $x(t)$ associated to tracker evolving according to \eqref{eq:nonl_system} satisfies $x(t) \in \mathcal{X}$ for all times $t \geq 0$ under the policy $\kappa^{\bar{\theta}^\star}$. 
\end{proposition}
\begin{proof}
Let \eqref{eq:theta_ub_sos} be feasible and Assumption~\ref{ass:mpc_feas} hold. Since, $\mathcal{O}^{\bar{\theta}^\star}$ is a forward reachable set of $\Omega$ under policy $\kappa^{\bar{\theta}^\star}$, we have $e(0) \in \Omega \implies e(t) \in \mathcal{O}^{\bar{\theta}^\star}$, for all $t \geq 0$. Therefore, feasibility of \eqref{eq:MPC_plan} guarantees $\hat{\mathcal{X}}^{\bar{\theta}^\star} \oplus \mathcal{O}^{\bar{\theta}^\star} \subseteq \mathcal{X}$, implying ${x}(t) \in \mathcal{X}$ for all $t \geq 0$.   
\end{proof}

\section{Numerical Example:  Double Pendulum}
In this section we present a numerical example with our proposed Algorithm~1 and Algorithm~2. 
For the fully-actuated double pendulum example from \cite{YinStan2019}, the polynomial dynamics obtained from a least-squares approximation for $(x_1, x_3) \in [-1,1] \times [-1, 1]$ are 
\begin{align}
\bmat{\dot{x}_1 \\ \dot{x}_2 \\ \dot{x}_3 \\ \dot{x}_4} &= \bmat{x_2 \\ f_2(x_1, x_2, x_3, x_4) \\ x_4 \\ f_4(x_1, x_2, x_3, x_4)} + \bmat{0 & 0 \\ 8 & -31.2 \\ 0 & 0 \\ -31.2 & 391.2}\bmat{u_1 \\ u_2}, \nonumber
\end{align}
\begingroup
\begin{align}
f_2 &=   - 3.447 x_1^3 + 2.350 x_1^2 x_3 + 1.303 x_1 x_3^2
  + 3.939 x_3^3 \nonumber \\
  & ~~~~~~~~~~~~~~~~~~~~~~~~~~~~~~~~ + 21.520 x_1 - 5.000 x_3, \nonumber \\
f_4 &=  4.023 x_1^3 - 36.551 x_1^2 x_3 - 4.131 x_2^2 x_3 - 27.060 x_3^3 \nonumber \\
& ~~~~~~~~~~~~~~~~~~~~~~~~~~~~~~~~ - 25.115 x_1 + 77.700 x_3, \nonumber
\end{align}
\endgroup
where $x_1$ and $x_3$ are angular positions of the first and second links (relative to the first link), $x_2$ and $x_4$ are angular velocities of the first and second links (relative to the first link), $u_1$ and $u_2$ are torques applied at the joint 1 and joint 2. The angular positions and applied torques are shown in Fig.~\ref{fig:angles}. The control objectives are: $(i)$ to bring $x$ from initialized $(-0.57, 0.52, 0, 0.02)$ to target $(0.3, 0, 0, 0)$ and maintain it there, and $(ii)$ to satisfy state constraints
\begin{align}\label{eq:trac_con_num}
\mathcal{X} := \{(x_1, x_2) : \vert x_1 \vert \leq 0.6, \ \vert x_2 \vert \leq 1.3\}. 
\end{align}

\subsection{Planner Parametrization} 
Based on the control objective, we use a single inverted pendulum as the low-fidelity model to generate planning trajectories $(\hat{x}_1(t), \hat{x}_2(t))$ for the planner. The polynomial dynamics of this low-fidelity planner are given as
\begin{align}
    \bmat{\dot{\hat{x}}_1 \\ \dot{\hat{x}}_2} = \bmat{\hat{x}_2 \\ -5.131 \hat{x}_1^3 + 32.1 \hat{x}_1} + \bmat{0 \\ 9.1}\hat{u}, \nonumber
\end{align}
where $\hat{x}_1$ represents the angular position of the single inverted pendulum (shown in Fig.~\ref{fig:angles}), $\hat{x}_2$ is the angular velocity, and $\hat{u}$ is the torque applied at joint 1.
\begin{figure}[h]
	\centering
	\includegraphics[width=0.25\textwidth]{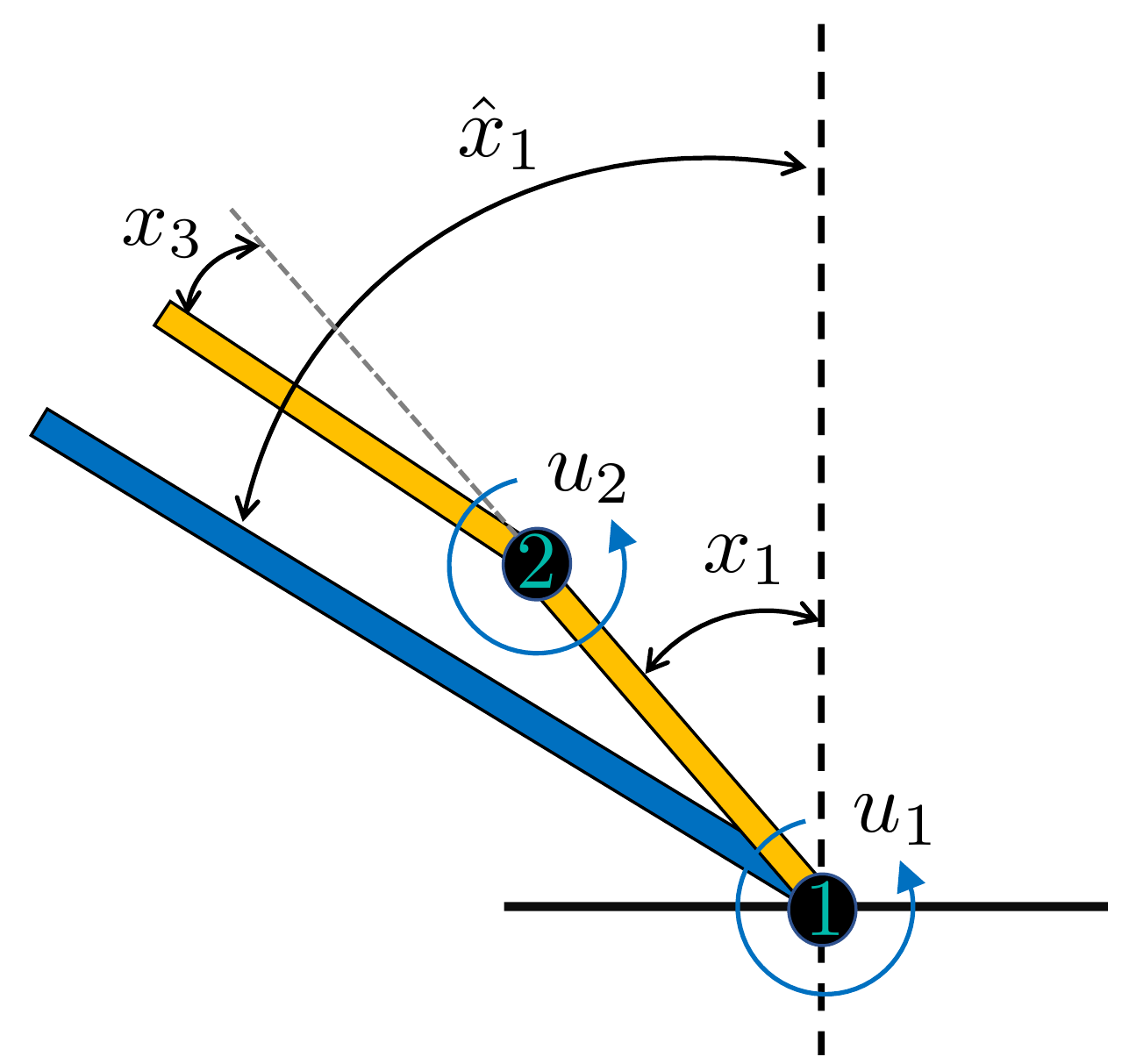}
	\caption{Double pendulum and its abstraction as a single pendulum. The angular positions are labelled.}
	\label{fig:angles}    
\end{figure}
We want $(x_1(t), x_2(t))$ to track $(\hat{x}_1(t), \hat{x}_2(t))$, while enforcing $(x_3(t), x_4(t))$ to stay close to the origin. Therefore, the map $\pi(\cdot)$ is chosen to be
\begin{align}
    \pi(\hat{x}) = P \hat{x}, \ \text{where} \ P = [\boldsymbol{I}_2, \boldsymbol{0}_{2 \times 2}]^\top. \nonumber
\end{align}
The constraint sets for the planner are parametrized by
\begin{subequations}\label{eq:plann_par_num}
\begin{align}
& \hat{\mathcal{U}}^\theta = \{\hat{u} \in \R : \vert \hat{u} \vert \leq 5 \},\\  & \hat{\mathcal{X}}^\theta = \{\hat{x} \in \R^2 : \vert \hat{x}_1 \vert \leq 0.6 \theta_1, \ \vert \hat{x}_2 \vert \leq 1.3 \theta_2\}, 
\end{align}
\end{subequations}
for all $ \theta = (\theta_1,\theta_2)$ such that $\theta_1 , \theta_2 \in [0, 1]$. Take the set of initial conditions for error-states as $\Omega = \{e \in \R^4:e_1=e_2=e_3=0, -0.03\leq e_4 \leq 0.03\}$.


\subsection{Parametric Error Bound Computation}
 In this example, $V$ is chosen to be a degree-2 polynomial in $(e, \theta)$, and $\kappa$ is chosen to be a degree-4 polynomial in $(e, \hat{x}, \hat{u}, \theta)$. The SOS optimizations in Algorithm~\ref{alg:alg1} are formulated using the sum-of-squares module SOSOPT \cite{Pete:13} on MATLAB. 
 and solved by Mosek \cite{Mosek:17}. 
After solving \eqref{eq:Vk_sosopt}, we obtain the parametric error bound $\mathcal{O}^\theta$ and the associated feedback controller for tracker, $\kappa^\theta$. 

\subsection{Optimal Planner-Tracker Design}
In this section we highlight the ``safety by design" aspect of Algorithm~\ref{alg:alg2}, as a consequence of solving \eqref{eq:theta_ub_sos}. Instead of fixing the planner constraint sets $\hat{\mathcal{X}}$ and $\hat{\mathcal{U}}$ heuristically as in \cite{YinStan2019}, we enmesh the planner-tracker design phases, looking for the best parameter $\bar{\theta}^\star$ in \eqref{eq:plann_par_num} that satisfies \eqref{eq:con_feasPar}. The inclusion of Algorithm~\ref{alg:alg2} inherently ensures safety (satisfaction of constraints \eqref{eq:trac_con_num} by tracker states $x(t)$ for all times $t$) by design, while simultaneously allowing for the maximum permissiveness of the planner in \eqref{eq:MPC_plan}. For the following simulations, we set $\hat{x}(0) = (-0.57, 0.52)$, i.e. $e(0)=(0,0,0,0.02)$. 

\subsubsection{Failure of Heuristics}
In search for the most permissive planner, the first planner design scenario involves setting $\hat{\mathcal{X}}^{\tilde{\theta}} = \mathcal{X}$ 
, i.e. $\tilde{\theta} = (1,1)$ in \eqref{eq:plann_par_num}.
As expected, the tracker can easily violate safety constraints \eqref{eq:trac_con_num}.
\textcolor{black}{For satisfying \eqref{eq:trac_con_num} by the tracker}, we next use our heuristics and set $\tilde{\theta} = (0.99,0.99)$ and $(0.98,0.98)$ in the next two cases respectively. We see from Fig.~\ref{fig:result_theta_99percent} and Fig.~\ref{fig:result_theta_98percent} that both $\mathcal{O}^{(0.99,0.99)}$ and $\mathcal{O}^{(0.98,0.98)}$ cross the safety constraints $\mathcal{X}$ given in \eqref{eq:trac_con_num}. Hence both the heuristic parameters are rendered invalid. In fact in Fig.~\ref{fig:result_theta_99percent}, we also see the tracker trajectory violating \eqref{eq:trac_con_num}.
\begin{figure}[h]
	\centering
	\includegraphics[width=1\columnwidth]{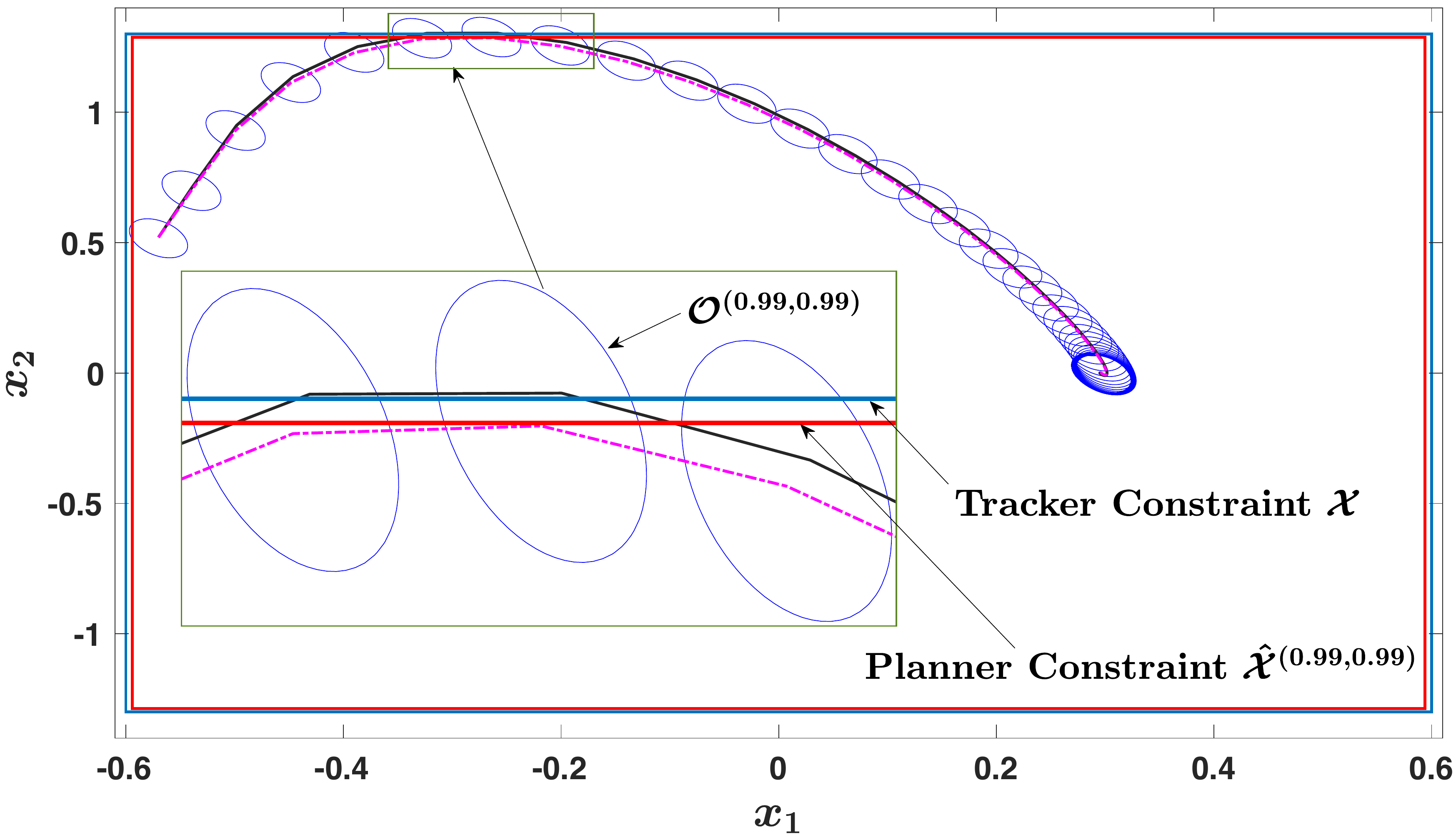}
	\caption{Planner design using $\tilde{\theta} = (0.99,0.99)$. Dashed purple curve denotes planner trajectory and solid black curve is corresponding tracker trajectory.}
	\label{fig:result_theta_99percent}  
\end{figure}
\begin{figure}[h]
	\centering
	\includegraphics[width=1\columnwidth]{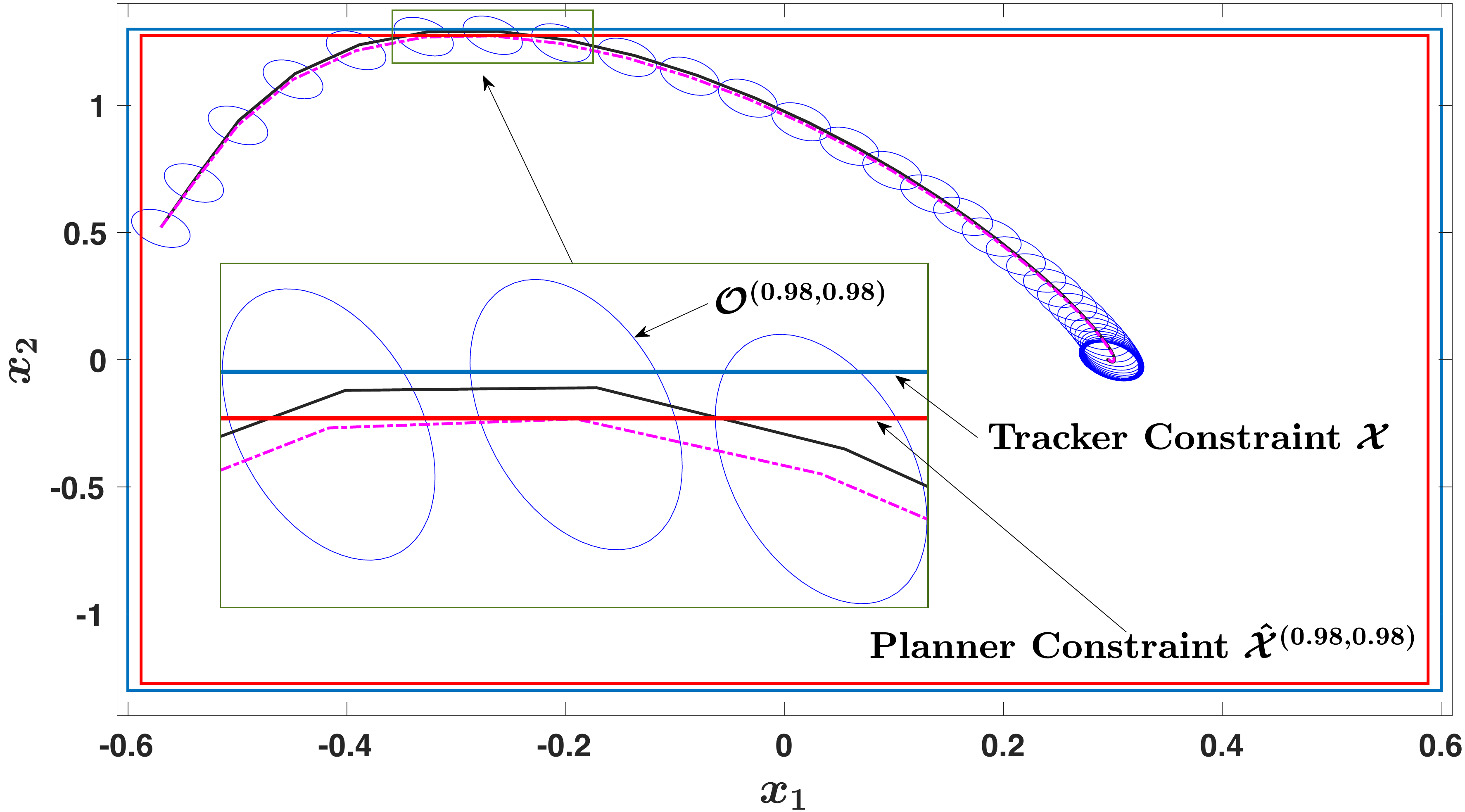}
	\caption{Planner design using $\tilde{\theta} = (0.98,0.98)$.}
	\label{fig:result_theta_98percent}  
\end{figure}

\subsubsection{Optimal Parametrization}
Using our Algorithm~\ref{alg:alg2}, the computed $\bar{\theta}^\star = (0.954, \ 0.940)$, i.e. the most permissive planner state constraint set is $\hat{\mathcal{X}}^{\bar{\theta}^\star} = \{\hat{x} \in \R^2 : \vert \hat{x}_1 \vert \leq 0.5724, \ \vert \hat{x}_2 \vert \leq 1.2220 \}$. In Fig.~\ref{fig:result_theta_star}, the planner uses $\hat{\mathcal{X}}^{\bar{\theta}^\star}$ as the state constraint. We can see that 
the error bound $\mathcal{O}^{\bar{\theta}^\star}$ around the planner trajectory remains within $\mathcal{X}$, which guarantees the safety of the tracker trajectory. The tracker trajectory never violates $\mathcal{X}$. This highlights that Algorithm~\ref{alg:alg2} provides safety guarantees, and enables the designer to avoid repeated planner-tracker design in search for safety. 
\begin{figure}[h]
	\centering
	\includegraphics[width=1\columnwidth]{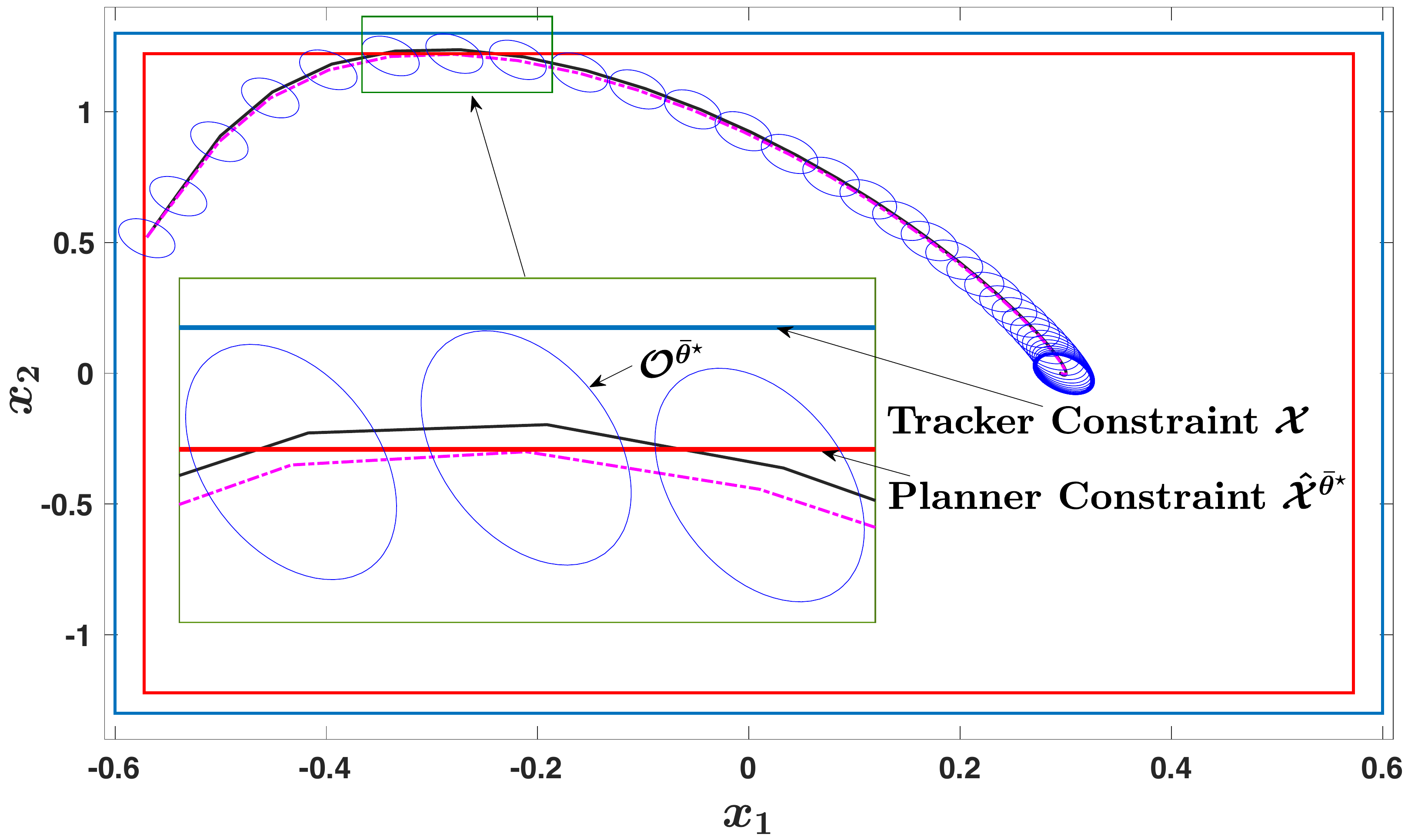}
	\caption{Planner and tracker design with optimal $\bar{\theta}^\star$. }
	\label{fig:result_theta_star}  
\end{figure}

\section{Conclusions}
    We presented an optimization based safe-by-design approach of trajectory planning--tracking for nonlinear systems. 
    Instead of heuristically picking the constraints imposed on the planner, we parametrized them with additional design parameters. Consequently, the tracking error bound and the tracking control law are parametrized too, and are computed through Sum-of-Squares programming (Algorithm~\ref{alg:alg1}). The optimal design parameters are chosen (Algorithm~\ref{alg:alg2}) specifically ensuring tracker safety along with maximum permissiveness of the planner. 
    \balance
\section*{Acknowledgements}
We thank professor Francesco Borrelli for providing helpful commetns. This work was supported in part by the grants ONR-N00014-18-1-2209, ONR-N00014-18-1-2833, AFOSR FA9550-18-1-0253, and NSF ECCS-1906164. 



\renewcommand{\baselinestretch}{0.91}
\bibliographystyle{IEEEtran}
\bibliography{IEEEabrv,bibliography.bib} 
\renewcommand{\baselinestretch}{1}


\end{document}